\documentclass[english]{article}
\usepackage[T1]{fontenc}
\usepackage{geometry}
\usepackage{hyperref}
\usepackage{babel}
\usepackage{mathrsfs}
\usepackage{amsmath}
\usepackage{amsthm}
\usepackage{mathtools,amssymb}
\usepackage{graphicx,graphics}
\usepackage{lineno}
\usepackage{epstopdf}
\usepackage{subcaption}
\newtheorem{theorem}{Theorem} 
\newtheorem{proposition}{Proposition}

\begin{document}
%\begin{linenumbers}
\begin{center}  {\Large \bf  Modeling the effects of information--dependent vaccination behavior on meningitis transmission}
\end{center}
\smallskip
\begin{center}
{\small \textsc{Bruno Buonomo\footnote{buonomo@unina.it}}, \textsc{Alberto d'Onofrio\footnote{alberto.donofrio@i-pri.org}},
 \textsc{Semu Mitiku Kassa\footnote{semu.mitiku@aau.edu.et}},
 \textsc{Yetwale Hailu Workineh\footnote{yetwale@yahoo.com}}}
\end{center}
\begin{center} {\small \sl $^1$ Department of Mathematics and Applications, University of Naples Federico II  via Cintia, I-80126\\ Naples, Italy}.
 \\
 {\small \sl $^{2}$ International Prevention Research Institute, 95 cours Lafayette, 69006 Lyon, France}\\
{\small \sl $^{3}$ Department of Mathematics and Statistical Sciences, Botswana International Science and Technology (BIUST), Palapye, Botswana}.\\
{\small \sl $^{3, 4}$ 
Department of Mathematics, Addis Ababa University, Addis Ababa, Ethiopia}.
\end{center}

\medskip

{\centerline{\bf Abstract} 
\begin{quote}
\small We propose a mathematical model to investigate the effects of information--dependent vaccination behavior on meningitis  transmission. The information is represented by means of information index as early proposed in (d'Onofrio et al., Theor. pop. biol., 2007). We perform a qualitative analysis based on stability theory, focusing to the global stability of the disease free equilibrium (DFE) and the related transcritical bifurcation taking place at the threshold for the DFE.\\
Finally, we assess the role of epidemiological and information parameters in the model dynamics through numerical simulations. 
Our simulations suggests that the impact of the human behavior critically depend on the average information delay. For example, it can induce recurrent epidemics, provided that transfer rate from the carrier to the infectious state is over a threshold. Otherwise, the endemic equilibrium is (at least) locally stable.
\end{quote}}


\normalsize

\section{Introduction}
\label{sec:intro}

Meningitis is an inflammation of the meninges, the protective membranes that surrounds the brain and spinal cord. Meningitis is mainly  caused by bacteria or viruses, although it may be rarely caused also by fungi, parasites or free-living microscopic ameba \cite{CDC}. Viral meningitis is more common  than bacterial meningitis but often less severe. This last can be deadly if not treated right away \cite{CDC}.

Meningococcal disease (MD) is a serious bacterial form caused by Neisseria meningitidis bacteria, which has the potential to cause large epidemics. It can cause severe brain damage and is associated with high fatality: it is fatal in 50\% of cases if untreated. This form of meningitis is observed worldwide but the highest burden of the disease is in the `meningitis belt', an area  of sub-Saharan Africa that stretches from Senegal in the west to Ethiopia in the east, where around 30 000 cases are reported each year \cite{WHO}.

Neisseria meningitidis bacteria are transmitted from person-to-person through respiratory droplets or throat secretions from carriers. The transmission of N. meningitidis is facilitated during mass gatherings (as pilgrimages and jamborees) \cite{WHO}. Most cases are acquired through exposure to asymptomatic carriers and relatively few through direct contact with patients with MD \cite{B2005}.  

There are twelve different types of N. meningitidis, called serogroups, that have been identified. Only six of them (A, B, C, W, X and Y) can cause epidemics \cite{WHO}. The  N. meningitidis serogroup A (NmA) is the most commonly isolated pathogen in the African meningitis belt \cite{FPS2016,TG2009, Wireko2012}, where it has historically accounted for 90\% of MD cases and the majority of large-scale epidemics \cite{CDC}. 

Mathematical modeling have the potential, and offer a promising way, to design effective prevention and control strategies. 
An early approach to modeling the spread of meningitis is due to Pinner et al.  in \cite{Pinner1992} where the weekly \emph{attack rate} (i.e. the number of individuals who get infected divided by the number of people at risk for the infection) was estimated for  both the vaccinated and non--vaccinated populations during an epidemic of NmA in Nairobi, Kenya, in 1989.

The age structure of population has been often considered as an important modeling aspect in models of meningococcal infection spread.
Martcheva and Cornell \cite{MC2003} used an age-structured meningitis model to identify the contribution of the carriers (infected who have the bacteria but appear healthy) to the transmission of MD. Tuckwell et al. \cite{THV2003} proposed a discrete time model including age-dependent rates to determine the impact of vaccination  schedule on the time course of a meningococcal serogroup C (NmC) disease in France. An age--structured model was developed by Trotter et al. \cite{TGE2005} and then fitted to data on immunization with NmC conjugate vaccines in England and Wales to investigate the effects of a conjugate vaccine program. Deterministic compartmental models were fitted to age structured data sets of MD by Coen et al. \cite{CCS2000}. 

In recent years, mathematical models for MD based on the Susceptible--Carrier--Infectious--Recovered (SCIR) structure have received much attention by scholars. Simpson and Roberts \cite{SR2012} used a SCIR model to estimate the long--term effects of the 2004 vaccination campaign on the epidemic of meningococcal B disease in New Zealand. A SCIR model  developed by Irving, Blyuss and coworkers \cite{Blyuss2016,IBC2012} (where they also developed some other SCIR-related models) suggested that temporary population immunity is an important factor to include in models designed to attempt to predict meningitis epidemics and to measure the efficiency of vaccines being deployed. G. Djatcha Yaleu et al. \cite{Yaleu2017} considered vaccination in a SCIR model for the dynamical transmission of NmA. They conclude that the control of the epidemic of NmA pass through a combination of a large coverage vaccination of young susceptible individuals and the production of a vaccine with a high level of efficacy. However, they used a set of parameters that are highly unrealistic. 

An important factor which has not been included in meningitis models, as far as we know, is the change of social behavior of individuals 
as consequence of information and rumors on the spread of the disease within their community \cite{MD2013,WBB2016}. This phenomenon is a key factor regulating the propensity of individuals in adopting non--mandatory protective measures. Well known examples are the vaccination for childhood diseases like measles, mumps and rubella \cite{WBB2016} and bed--nets for mosquito--borne diseases \cite{asbbsk}. Also in the case of meningitis, there is evidence that the perceptions of individuals about the disease affect their decisions regarding therapeutic and preventive interventions \cite{coetal}. This, in turn, influence disease occurrence, morbidity and mortality. Specific studies regarding vaccination against Meningococcal C revealed that parents' perceived vulnerability and perceived control in preventing the infection seem to influence parents' evaluation of the vaccination programme. Therefore,  these perceptions may, in principle, affect vaccination behavior and could be relevant to be taken into account when educating population about vaccination \cite{tietal}. Another interesting investigation is the recent study focused on evaluating the knowledge and attitudes about Meningococcal B and the relative vaccine for children among a sample of parents in Italy \cite{MNAD2017} (see also a similar survey in Poland \cite{dretal}). One of the conclusions, only apparently trivial, was that people who knew that the vaccine was a preventive measure of meningitis were more likely to have a positive attitude towards vaccination. On the other hand, in spite of the great successes obtained by immunization programmes (as the long-term strategies in the African meningitis belt with vaccine \emph{MenAfriVac}\textregistered$\,$  against NmA meningitis \cite{Lafoetal}) vaccine hesitancy threatens to erode these gains \cite{coopetal,Lafoetal} and push health policy makers to plan specific interventions \cite{ECDC}.

The feedbacks of human behavior on the transmission and control of infections is the main focus of the recently emerging research field of behavioral epidemiology (BE) of infectious diseases \cite{MD2013,WBB2016}. Indeed, in classical epidemic models \cite{anma, cap} individuals are modeled as passive particles randomly interacting according to the mass-action principle of statistical physics \cite{WBB2016}, as they were reacting molecules of chemistry. This paradigm does not account the role of human decisions in influencing the spread and the control of infectious diseases, which are instead explicitly modeled in the framework of BE \cite{vewibe}.

Vaccines for meningitis like \emph{MenAfriVac} are given in mass for ages usually between 1 and 29 \cite{LRD2009,SYG2018}. A possible way to modeling behavior--dependent vaccination at all ages relies on adopting  the information--based  approach,  which focuses on the case where the driving force of the vaccine demand is represented by the time changes in the perceived risk of infection. This approach has been introduced in \cite{DMS2007,DMS2007b} by mean of a phenomenological model where the vaccination rates of newborn is an increasing function of the information acquired on the spread of the infectious disease. Main applications concerned vaccines targeted at childhood infectious diseases such as measles, mumps and pertussis \cite{BDL2008,BDL2013,DMS2007,DMS2007b,MD2013,WBB2016}. Only few applications, as far as we know, concern with the case of vaccination at all ages  \cite{BDL2013,bbrdm,DMS2007b,kumar}.

In this paper, motivated by the SCIR model developed in \cite{IBC2012} and its variants \cite{Yaleu2017,SR2012}, we propose a BCM mathematical model of meningitis transmission, including information--dependent vaccination at all ages. The main aim is to theoretically assess, through the proposed model, how the information may affect the dynamics of meningitis transmission within a given population. We will make use of analytical methods for dynamical systems, like stability and bifurcation theory. The theoretical evaluation of the effects of human decision on the disease control will be complemented by means of numerical simulations.

The paper is organized as follows: The model is presented in Section \ref{sec: model}. In Section \ref{sec: basic}, the basic properties of the model are given as well as the determination of the basic reproduction number. We also show that an endemic equilibrium is possible and it is unique. In Section \ref{sec: stability} we provide a stability analysis of the disease free equilibrium and show, through bifurcation analysis, that the endemic equilibrium is stable, at least when the basic reproduction number is close to its critical value. In Section \ref{sec:sens} numerical simulations are illustrated. Conclusions are given in  Section \ref{sec: concl}.

%\newpage
\section{The model}
\label{sec: model}

We assume that a given population may be divided into five distinct compartments: the susceptible individuals $(S)$, who are susceptible to infection; the vaccinated individuals $(V)$; the carriers $(C),$ who are carrying the infection and are infectious, but show no signs of the invasive disease; the infectious  $(I)$, who have been infected by the disease and become immediately infectious; the recovered  individuals $(R)$. The total population is denoted by $N$. Therefore, at time $t$ it is 
$$
N(t) = S (t) + V (t) + C (t) + I (t) + R (t).
$$
Individuals are assumed to be recruited into the population at a constant rate $\Lambda$. A mass vaccination for all susceptible individuals is given with rate $\nu$, so that $\nu S$ individuals are transferred to the vaccinated group $V$. We assume that  vaccination is lifelong and does not have a waning effect.\\
After making an effective contact with carriers or infectious, the susceptible individuals become carriers. The force of infection \cite{cap} of the disease, denoted as $\lambda$, will be described later. Carriers becomes infectious at a rate $\sigma$ and recover at a rate $\delta$. Infectious individuals recover at a rate $\rho.$  The model takes into account both the natural death rate $\mu$ and disease-induced mortality rate $d$. Recovery from the disease is not permanent, so we assume that the recovered individuals become susceptible again at rate $\phi$. 

Therefore, the baseline model is  governed by the following system of non linear ordinary differential equations:
\begin{equation}
\label{system1}
\begin{array}{l}
 \dot S = \Lambda-\lambda(N) S -\nu S+\phi R-\mu S\\
 \dot V = \nu S-\mu V\\
 \dot C = \lambda(N) S-(\sigma+\delta+\mu)C\\
  \dot I =\sigma C-(\rho+\mu +d)I\\
 \dot R =\delta C+\rho I-(\phi+ \mu)R,
  \end{array}
\end{equation}
where the upper dots denote the time derivatives. As for the force of infection $\lambda(N)$, due to the asymptomatic nature of carriers, their contact rate is supposed to be greater than that of infectious, which are subject to identification and hospitalization. Therefore 
the level of infectiousness of $C$ is greater than $I$, say, by  $\epsilon>1$ factor \cite{ANS2018}. As a consequence, the force of infection is given by 
\begin{equation}
\label{foi}
\lambda(N)=\beta\dfrac{\epsilon C+I}{N},
\end{equation}
where $\epsilon\geq 1$, and $\beta$ is the transmission rate of infectious. \\
Note that if $d=0$ (and, of course, in absence of the disease) then the steady state value of the population is $\tilde N = \Lambda/\mu$, since the dynamics of $N(t)$ is ruled by
\begin{equation}
\label{eqNnod}
\dot N=\Lambda-\mu N.
\end{equation}

Now, we will build a behavioral variant of model (\ref{system1}) by introducing an information--dependent vaccination. Following the approach of information-dependent models \cite{DMS2007,WBB2016} we assume that the vaccination rate $\nu$ is an increasing function of an information index $M$, which summarizes the information about the current and past values of the disease \cite{DMS2007,WBB2016}:
\begin{align}
\label{M}
M(t)=\int_{-\infty}^{t} g(C(\tau),I(\tau))K_{a}(t-\tau) d\tau.
\end{align}
Here the function $g$ describes the role played by the infectious size in the information dynamics and is assumed to be dependent on the prevalence, i.e. on the compartments $C$ and $I$ (the basic properties of $g$ will be specified later). The term  $K_{a}$ is the delay kernel function, which represents the weight given to past prevalence. The information variable $M$ is described by a distributed delay for two main reasons. On one hand, people may have memory of past epidemics: indeed vaccine--related decision seldom depends on the current state of the spread of the disease. On the other hand, a delay in people awareness may be consequence of the time-consuming routine procedures (like clinical tests, notification of cases, reporting delays to public health authorities, etc.) that precede the dissemination of information on the disease's status in the community.

Combining \eqref{system1} and \eqref{M},  model (\ref{system1}) is modified in the following  nonlinear integro--differential system :
\begin{equation}\label{system2}
\begin{array}{l}
\dot S = \Lambda-\lambda(N) S -\nu (M)S+\phi R-\mu S\\
 \dot V = \nu(M) S-\mu V\\
 \dot C = \lambda(N) S-(\sigma+\delta+\mu)C\\
\dot I =\sigma C-(\rho+\mu +d)I\\
 \dot R =\delta C+\rho I-(\phi+ \mu)R\\
 \hspace{-0.5cm} M(t) =\int_{-\infty}^{t} g(C(\tau),I(\tau))K_{a}(t-\tau) d\tau,
  \end{array}
\end{equation}
where $\lambda(N)$ is given in (\ref{foi}) and the function $\nu(M)$ models the information-dependent rate of vaccinations, and it may be  split as follows:
\begin{equation}\label{nuM}
\nu(M)=\nu_{0}+\nu_{1}(M),
\end{equation}
where $\nu_{0}$ is a positive constant representing the fraction of susceptibles that are vaccinated independently on the available current and historical information on the prevalence level of the disease and $\nu_{1}(M)$ models the fraction of susceptibles that are vaccinated depending on the social alarm caused by the disease \cite{DMS2007}.

The following basic assumptions are done regarding the functions $g(C,I)$ and $\nu_1(M)$: \emph{(i)} $g(C,I)$ and $\nu_{1}$ are continuous and differentiable, except in some cases, at finite number of points;  \emph{(ii)} $g(0,0) = 0$, and $\nu_{1}(0) = 0$;
\emph{(iii)} $g(C,I)$ and $\nu_{1}$ are positive for any positive value of their respective arguments; \emph{(iv)} $0 < \partial g/\partial C<c_{1}$, $0< \partial g/\partial I <c_{2}$, and $0 <\nu'_{1}(M)<c_{3}$, for some constants $c_{1}$, $c_2$ and $c_{3}$.

The simplest specific functional form of $\nu_{1}(M)$ is the linear function, 
\begin{equation}\label{nu1}
\nu_{1}(M)=b\,M, 
\end{equation}
where $b$ is a positive constant. Moreover, since the possibility of getting  information from carriers group $C$ is almost not possible, we take $g$ as a function of only the infected groups $I$. Again, the simplest form is:
\begin{equation}\label{glin}
g(I)=h\,\frac{I}{\tilde{N}}=k I.
\end{equation}
Note that the size of the infectious compartment is normalized by divinding $I$ by the reference steady state population. 
In this case, the positive constant parameter $h$ can be interpreted as the \emph{information coverage}. This may be seen as a `summary' of two contrasting phenomena:  the phenomenon of disease under--reporting and the level of media and rumours coverage of the status of the disease, which tends to amplify the social alarm \cite{BDL2008}.

As far as the kernel $K_{a}(t)$ is concerned, a relevant example is the weak exponential delay kernel \cite{MM2008}
\begin{equation}\label{expker}
K_{a}(t) = a\,e^{-at},
\end{equation}
where the parameter $a$ assumes the biological meaning of inverse of the average delay of the collected information on the disease, as well as the average length of the historical memory concerning the disease in study. The Kernel (\ref{expker}) is a particular case of 
Erlangian distribution. This choice has the advantage of making the corresponding integro-differential system reducible
into ordinary differential equations through the \emph{linear chain trick} \cite{MM2008}.
Model (\ref{system2})--(\ref{foi})--(\ref{M}), with choices (\ref{nuM}), (\ref{nu1}), (\ref{glin}) and (\ref{expker}), may be written:

\begin{equation}\label{system4}
\begin{array}{l}
 \dot S=  \Lambda-\lambda(N) S -(\nu_{0}+bM)S+\phi R-\mu S\\
 \dot V=  (\nu_{0}+bM )S-\mu V\\
 \dot C=  \lambda(N) S-(\sigma+\delta+\mu)C\\
\dot I= \sigma C-(\rho+\mu +d)I\\
 \dot R= \delta C+\rho I-(\phi+ \mu)R\\
  \dot M= a\, \left(kI-M\right),
\end{array}
\end{equation}
and will be the subject of the investigation in this paper. The  parameters, their descriptions and the baseline values are given in  Table \ref{parameters2}.

\begin{table}[t]
\resizebox{\textwidth}{!}{%
\begin{tabular}{|l|l|l|l|l|}
\hline{\small Parameter}  & {\small Description} & {\small Value (range)}& {\small Unit}& {\small Source and comments}\\
\hline
$\mu$  & {\small Per--capita natural mortality rate for all classes} & 0.027397 & $years^{-1}$&Life expectancy is $50$ years \cite{IBC2012}\\
$\tilde{N}$  & 
{\small Steady state Population} &1000000 & adim. & Assumed\\
$\Lambda$  & 
{\small Inflow rate of susceptible individuals} & 20000 & $years^{-1}$ & $=\mu \tilde{N}$\\
$\beta$  & {\small Transmission Rate} & 100 $(50,200)$ & $years^{-1}$& \cite{IBC2012}\\
$\rho$  & {\small Rate of recovery of infectious individuals} & 52 &$years^{-1}$&Disease duration $\approx 1$ week  \cite{IBC2012}\\
$d$  & {\small Disease--induced death rate} &5.2 & $years^{-1}$ & Mortality probability $\approx 0.1$ \cite{IBC2012}\\
$\sigma$  & {\small Rate moving from carrier to infected} &26 (0.1,52)& $years^{-1}$ &\cite{IBC2012}\\
$\delta$  & {\small Recovery rate of carriers} &26 $(0.1,52)$ & $years^{-1}$ &\cite{IBC2012}\\
$\phi$ & {\small Rate of loss of immunity} & 1 $(0.04,2)$ & $years^{-1}$&\cite{IBC2012}\\
$\nu_{0}$ & {\small behavior--independent vaccination rate} &0.04 & $years^{-1}$ &Assumed\\
$T$ & {\small Average information delay}  & $[0,120]$ & $days$ & From no delay to 4 months about. \\
$a$ & {\small Inverse of average information delay}  & $=365.25/T_{deal}$ & $years^{-1}$ & Inverse of $T$ in $years^{-1}$\\
$b$ & {\small Slope of vaccination rate w.r.t. information index $M$} &2000 $(1,\infty)$ & $years^{-1}$  &Assumed\\
$h$ & {\small Information coverage}  & 0.5 $(0.1,1)$ & adim. &Assumed\\
$\epsilon$  & {\small Enhanced infectiousness of carriers} & 1.3  & adim. &Assumed\\
\hline
\end{tabular}
}
\caption{Description and baseline values of the parameters of model (\ref{system4}).}
\label{parameters2}
\end{table}

\section{Basic properties}
\label{sec: basic}

We begin by determining the biologically feasible set for model (\ref{system4}). The following proposition also implies that the solutions of  (\ref{system4}) are bounded, provided that the initial conditions are non negative.

\begin{theorem}
The set
\begin{equation}
\label{region}
\Omega=\{(S,V,C,I,R,M)\in \mathbf{R}^6_+:0\leq M\leq k \tilde{N},\;\; 0\leq \tilde{N} \frac{\mu}{\mu+d} \leq S+V+C+I+R \leq \tilde{N}\}.
\end{equation}
is positively invariant and absorbing. 
\end{theorem}
\begin{proof} It suffices to note that adding the first five equations of  \eqref{system4}, it follows
\begin{equation}
\label{eqN}
\dot N=\Lambda-\mu N-dI,
\end{equation}
so that from
$$ \Lambda-(\mu + d) N \leq \dot N \leq \Lambda-\mu N $$
it follows 
$$
\limsup_{t\rightarrow\infty} N(t) \leq \tilde{N}.
$$
$$
\liminf_{t\rightarrow\infty} N(t) \leq \frac{\Lambda}{\mu+d}=\tilde{N}\frac{\mu}{\mu+d}.
$$
Then, the rest of the proof follows from standard arguments as in \cite{BDL2008}.
\end{proof}

It is easy to check that model (\ref{system4}) admits the \emph{disease--free equilibrium} (DFE), which represents the scenario where there is no infection within the community, given by:
\begin{equation}
\label{DFE}
E_{0} = \left(S_{0},V_{0}, 0, 0, 0, 0\right),
\end{equation}
where 
$$
S_0=\frac{\Lambda}{\mu +\nu_{0}};\;\;\;\;
V_0=\frac{\nu_{0}\Lambda}{\mu(\mu +\nu_{0})}
$$

Our next step is to determine the basic reproduction number (BRN) $\mathcal{R}_{0}$ associated to model (\ref{system4}). To this aim, we apply the \emph{next generation matrix method} (NGM) \cite{Driessche2002}. According to the well--known procedure \cite{Driessche2002}, we identify $(C, I)$ and $(S,V,R)$ as  the infected and uninfected classes, respectively. Then, we can write the matrix $\mathcal{F}$ for the new infection and the matrix $\mathcal{V}$ for the remaining transfer. More precisely, the matrix $\mathcal{F}$, where the entries represent the rates of new infections (i.e. the rate of appearance of new infections in compartments $C$ and $I$), here is a column vector, actually, given by:
\begin{equation*}
\mathcal{F}=\left[\begin{array}{c}
\lambda S\\ 0 
\end{array}  \right]=\left[\begin{array}{c}
\frac{\beta}{N}(\epsilon C+I)S\\ 0 
\end{array}  \right].
\end{equation*}  
The matrix $\mathcal{V}$ where the entries represent the rate of infections transferred (i.e. it incorporates the remaining transitional terms, namely births, deaths, disease progression, and recovery in $C$ and $I$), is also a vector, here, given by:
\begin{equation*}
\mathcal{V}=\left[\begin{array}{c}
(\sigma +\delta +\mu)C \\ 
-\sigma C+ (\rho+\mu+d)I 
\end{array}  \right].
\end{equation*}
According to NGM, the BRN $\mathcal{R}_{0}$ is the greatest of the eigenvalues of the matrix $FV^{-1}$, where $F$ and $V$ are the matrices obtained by differentiating $\mathcal{F}$ and $\mathcal{V}$ with respect to $C$ and $I$ and then evaluated at the DFE $E_0$. In our case, we have:
$$
\tilde F=\left[\begin{array}{cc}
\frac{\beta S(\epsilon N-\epsilon C-I)}{N^2} & \frac{\beta S( N-\epsilon C-I)}{N^2} \\ 
0  & 0
\end{array}  \right],\;\;\;\;
\tilde V=\left[\begin{array}{ccc}
(\sigma+\delta +\mu )  & 0 \\ 
 -\sigma & (\rho+\mu+d)     
\end{array}  \right],
$$
and therefore,
$$
F=\tilde F({E_{0}}) =\left[\begin{array}{cc}
\frac{\beta \epsilon \mu}{\mu +\nu_{0}}  & \frac{\beta  \mu}{\mu +\nu_{0}} \\ 
  0   & 0
\end{array}  \right],\;\;\;
V=\tilde V({E_{0}})=\left[\begin{array}{cc}
(\sigma+\delta +\mu )  & 0 \\ 
 -\sigma & (\rho+\mu+d)     
\end{array}  \right].
$$
It can be easily checked that
$$
FV^{-1}=\left[\begin{array}{cc}
\frac{\beta\epsilon\mu }{(\mu +\nu_{0})(\sigma+\delta +\mu )}+\frac{\beta\mu\sigma}{(\mu +\nu_{0})(\sigma+\delta +\mu )(\rho+\mu+d)} &~~ 
\frac{\beta\mu}{(\rho+\mu+d)} \\ 
  0 & ~~0
\end{array}  \right],
$$
and hence the BRN is given by:
\begin{equation}\label{R0}
\mathcal{R}_{0}=\frac{\beta\epsilon\mu }{(\mu +\nu_{0})(\sigma+\delta +\mu )}+\frac{\beta\mu\sigma}{(\mu +\nu_{0})(\sigma+\delta +\mu )(\rho+\mu+d)}.
\end{equation}

The existence and uniqueness of an endemic equilibrium, which represent a steady state prevalence of the disease in the population, is given in the following proposition:

\begin{proposition}\label{propEE}
Model (\ref{system4}) admits an unique endemic equilibrium for $\mathcal{R}_{0}>1$.
\end{proposition}
\begin{proof}
Let us denote the generic equilibrium (i.e. the vector solution with positive constant components) as $E^* =(S^{*},V^{*},C^{*},I^{*},R^{*},M^{*})$. From (\ref{system4}) it follows:
$$
\begin{array}{c}
S^{*}=\frac{(\sigma+\delta+\mu)(\rho+\mu+d)}{\sigma \lambda^*}I^{*},\;\;\;\;
V^{*}=\frac{(\nu_{0}+bM^*)S^{*}}{\mu}=\frac{(\nu_{0}+bkI^*)(\sigma+\delta+\mu)(\rho+\mu+d)}{\mu\sigma \lambda^*}I^{*},\\
C^{*}=\frac{(\rho+\mu+d)}{\sigma}I^{*},\;\;\;R^{*}=\frac{\delta(\rho+\mu+d)+\sigma\rho}{\sigma(\phi +\mu)}I^{*},\;\;\;
M^{*}=kI^{*}.
\end{array}
$$
where
$$\lambda^* = \beta\dfrac{\epsilon C^*+I^*}{N^*},
$$
and $I^*$ is the positive solution of 
\begin{equation}\label{Iee}
b_{2}I^{*2}+b_{1}I^{*} +b_{0}=0,
\end{equation}
where
$$
\begin{array}{c}
b_{2}=\sigma bkd\,b_3, \\ \\
b_{1}= \sigma d(\nu_{0}+\mu)\,b_3-\Lambda\sigma bk\,b_3 + \beta\mu\phi\left[\delta\left(\rho+\mu+d\right)+\sigma\rho\right]\left[\epsilon \left(\rho+\mu+d\right) +\sigma\right]-b_3\beta\mu\left[\epsilon \left(\rho+\mu+d\right) +\sigma\right],\\ \\
b_{0}=\sigma\Lambda (\nu_{0}+\mu)\,b_3[\mathcal{R}_{0}-1],
\end{array}
$$
and
$$
b_3=(\phi+\mu)(\sigma+\delta+\mu)(\rho+\mu+d).
$$
Let 
\begin{equation}
f(I^*)=b_{2}I^{*2}+b_{1}I^{*} +b_{0}.
\end{equation}
Note that $f''(I^*)=b_{2}>0$, and $f(0)=b_{0}>0$ for $\mathcal{R}_{0}>1$. Moreover, since $N^*>0$, from (\ref{eqN}) it follows $I^*<\frac{\Lambda}{d}$, and direct calculations reveal that
$$
f\left(\frac{\Lambda}{d}\right)=-\beta\mu\frac{\Lambda}{d}\left[\epsilon \left(\rho+\mu+d\right)\sigma\right]\left\{\mu\delta\rho+\mu\sigma\rho+(\phi+\mu)\mu\left[(\rho+\mu+d)+\sigma\right]+\delta\mu^2+\mu\delta d\right\} <0.
$$
Therefore the function $f$ has only one zero on $\left[0,\frac{\Lambda}{d}\right]$, implying that the endemic equilibrium is unique provided that $\mathcal{R}_{0}>1$, whereas and no endemic equilibria are admissible for $\mathcal{R}_{0}<1$. 
\end{proof}

\section{Stability of equilibria}
\label{sec: stability}
We begin with the following local stability result:
\begin{proposition}
If $\mathcal{R}_{0}<1$ then $E_{0}$ is locally asymptotically  stable in $\Omega$, and unstable if $\mathcal{R}_{0}>1$.
\label{thmLAS}
\end{proposition}
\begin{proof}
This result is a direct consequence of the NGM used in Section \ref{sec: basic}, see Theorem 2 in \cite{Driessche2002}.
\end{proof}

In epidemiological terms, Theorem \ref{thmLAS} states that it is possible to control the epidemic if we can reduce the value of $\mathcal{R}_{0}<1$ as long as the initial population is in the neighborhood of the DFE point $E_{0}$. However, it is not difficult to show that for $\mathcal{R}_{0}<1$ the elimination of the disease happens independently from the initial size of the population. In other words, $E_{0}$ is globally asymptotically stable (GAS) for $\mathcal{R}_{0}<1$. This is stated in the following theorem.

\begin{proposition}
If $\mathcal{R}_{0} < 1$, then $E_{0}$ is globally asymptotically stable in $\Omega$.
\label{thmGAS}
\end{proposition}
\begin{proof} This result can be obtained by checking that all the (five) hypotheses of Theorem 4.3 in \cite{KS2008} are satisfied. This procedure is well known and used elsewhere (see for example, Theorem 2 in the recent paper \cite{koetal}). For our model the result can be obtained in the same way. Therefore, for the sake of brevity, we omit the details.
\end{proof}

\textbf{Remark.} \textit{Note that in the reality,it is unlikely to have baseline vaccination rates sufficiently large to eradicate the disease. In other words the above theorem is equivalent to say that under voluntary information-driven vaccination the eradication is impossible. }

As far as the stability of the endemic equilibrium $E^*$ is concerned, it is difficult to obtained an analytical stability result, even through the linearization method. However, we can use a bifurcation theory approach to get an insight on the stability properties of the model (and therefore of $E^*$) for values of the BRN close to the critical value $\mathcal{R}_{0}=1$. We know that the DFE $E_0$ is a \emph{nonhyperbolic} equilibrium for $\mathcal{R}_{0}=1$ (this follows because the Jacobian matrix of system (\ref{system4}), evaluated at $E_0$, admits a zero eigenvalue) and we also know from Proposition \ref{propEE} that there is an endemic equilibrium branch bifurcating from $E_0$ at $\mathcal{R}_{0}=1$. What we are interested in is to check that such a branch is stable, at least when BRN is close to the critical value $\mathcal{R}_{0}=1$ (in such case, the bifurcation is a \emph{transcritical forward bifurcation} \cite{CChavez2004}).  To this aim, we study the centre manifold at the \emph{criticality} (i.e. at $E_0$ for $\mathcal{R}_{0}=1$) by using the approach developed in \cite{CChavez2004,duetal,Driessche2002}, which establishes that the normal form representing the dynamics of the system on the centre manifold is given by
$$
\dot y=\,A_1\,\phi\, y+\,A_2\,y^2,
$$
where $\phi$ denotes a bifurcation parameter to be chosen, and
\begin{equation}
\label{Coeff_b}
\begin{array}{l}
A_1=v \cdot D_{x\phi}f(x_{0},0)u \equiv \sum\limits_{k,i,j=1}^n v_{k}u_{i}\frac{\partial^2 f_{k}(x_{0},0)}{\partial x_{i}\partial \phi}.
\end{array}
\end{equation}
and,
\begin{equation}
\label{Coeff_a}
\begin{array}{l}
A_2=\frac{v}{2}\cdot D_{xx}f(x_{0},0)u^2 \equiv \frac{1}{2}\sum\limits_{k,i,j=1}^n v_{k}u_{i}u_{j}\frac{\partial^2 f_{k}(x_{0},0)}{\partial x_{i}\partial x_{j}},
\end{array}
\end{equation}
where $f_{k}$, $k=1,\dots,n$ denote the right hand sides of the system \eqref{system4}, $x$ denotes the state vector, $x_{0}$ the disease free equilibrium $E_{0}$ and $v$ and $u$ denote the left and right eigenvectors,  respectively, corresponding to the zero eigenvalue of the Jacobian matrix of system  \eqref{system4} evaluated at the criticality.

In our case, let us choose $\beta$ as bifurcation parameter. Observe that $\mathcal{R}_{0}=1$ is equivalent
to
\begin{equation}
\label{betastar}
\beta=\beta^*:=\frac{(\mu +\nu_{0})(\sigma+\delta +\mu )(\rho+\mu+d)}{\mu\left[\epsilon(\rho+\mu+d)+\sigma\right]},
\end{equation}
so that the DFE $E_0$ is (globally) stable when $\beta<\beta^*$, and is unstable when $\beta>\beta^*$. Therefore, $\beta^*$ is a bifurcation value.

The direction of the bifurcation occurring at $\beta=\beta^*$ can be derived from the sign of coefficients (\ref{Coeff_a}) and (\ref{Coeff_b}). More precisely, if $A_1>0$ and $A_2<0$, then at $\beta_1=\beta^*$ there is a forward bifurcation \cite{CChavez2004}\\
In our case we have the following:

\begin{theorem}
\label{Forward}
System \eqref{system4} exhibits a forward bifurcation at $E_0$ and $\mathcal{R}_{0}=1$.
\end{theorem}

\begin{proof} The Jacobian matrix of system \eqref{system4} evaluated at $E_0$ and $\beta=\beta^*$ is given by:
\begin{equation}
\label{jbeta}
J_{\beta^*}=\left[\begin{array}{cccccc}
-(\nu_{0} +\mu)&0 & -\frac{\beta^* \epsilon \mu}{\nu_{0} +\mu} &-\frac{\beta^* \mu}{\nu_{0} +\mu}&\phi&-\frac{b\Lambda}{\mu+\nu_{0}}\\ 
\nu_{0}  & -\mu &0 & 0 & 0&\frac{b\Lambda}{\mu+\nu_{0}}\\
0&0 &\frac{\beta^* \epsilon \mu}{\nu_{0} +\mu}-(\sigma+\delta+\mu) & \frac{\beta^*  \mu}{\nu_{0} +\mu} & 0& 0\\
0 & 0 & \sigma & -(\rho +\mu + d) &0&0 \\
0 & 0 &\delta & \rho & -(\phi +\mu)&0 \\
0 & 0 &0 & ak &0&-a
\end{array}  \right]
\end{equation}
The eigenvalues are
$\lambda_{1} =-(\nu +\mu)$, $\lambda_{2}=-\mu$, $\lambda_{3}=-(\phi +\mu)$, $\lambda_{4}=-a$, and two roots of the quadratic equation:
$$\varphi^2+\left[\frac{\beta^* \epsilon \mu}{\nu_{0} +\mu}-(\sigma +\delta +\mu)-(\rho +\mu +d)\right]\,\varphi-\frac{\beta^* \mu}{\nu_{0} +\mu}(\epsilon(\rho +\mu +d)+\sigma)+(\sigma +\delta +\mu)(\rho +\mu +d)=0.
$$
From \eqref{betastar} it follows that
\begin{align*}
\varphi^2+\left\{\frac{(\sigma +\delta +\mu)(\rho +\mu +d)\epsilon}{\epsilon(\rho +\mu +d)+\sigma}-\left[(\sigma +\delta +\mu)+(\rho +\mu +d)\right]\right\}\varphi=0.
\end{align*}
Therefore, the remaining two eigenvalues are $\lambda_{5}=0$, and
$$
\lambda_{6}=-\frac{(\sigma +\delta +\mu)\sigma +(\rho +\mu +d)(\epsilon(\rho +\mu +d)+\sigma)}{\epsilon(\rho +\mu +d)+\sigma}<0.
$$
Hence, the matrix $J_{\beta^*}$ has zero as a simple eigenvalue and all the other eigenvalues are real and negative. 

Now, we determine the right and left eigenvectors of $J_{\beta^*}$ corresponding to the zero eigenvalue.
Let $u=(u_{1},u_{2},u_{3},u_{4},u_{5},u_{6})^T $ denote a right eigenvector associated with the zero eigenvalue $\lambda_{5}=0$, i.e.
$$
J_{\beta^*}u=0.
$$
After some algebraic manipulations, we get:
\begin{align*}
u_{1}&=\frac{\beta^* }{(\nu_{0} +\mu)^2}\left[
 -\frac{\epsilon \mu(\rho +\mu + d)}{\sigma}- \mu +\frac{\phi(\nu_{0} +\mu)[\delta(\rho +\mu + d)+\sigma\rho]}{\beta^*\sigma(\phi +\mu)} -\frac{b k\Lambda}{\beta^*}\right]u_{4},
\\ 
u_{2}&= [\frac{\nu_{0}}{\mu}\frac{\beta^* }{(\nu_{0} +\mu)^2}\left[
 -\frac{\epsilon \mu(\rho +\mu + d)}{\sigma}- \mu +\frac{\phi(\nu_{0} +\mu)[\delta(\rho +\mu + d)+\sigma\rho]}{\beta^*\sigma(\phi +\mu)} -\frac{b k\Lambda}{\beta^*}\right] +\frac{b k\Lambda}{\mu(\mu+\nu_{0})}]u_{4}, \\
u_{3}&= \frac{(\rho +\mu + d)}{\sigma}u_{4},
 \\
u_{4}&=u_{4} \geq 0,\\
u_{5}& = \frac{\delta(\rho +\mu + d)+\sigma\rho}{\sigma(\phi+\mu)}u_{4},\\
u_{6}&=ku_{4}.
\end{align*}
Similarly, the left eigenvectors can be determined from equation $vJ_{\beta^*}=\lambda v$, or equivalently $J_{\beta^*}^T v^T=0$, and we get:
\begin{align*}
v_{1}=v_{2}=v_{5}=v_{6}=0,\;\;
v_{3}= \frac{(\rho +\mu + d)}{\sigma}v_{4},\;\;
v_{4}=v_{4} \geq 0. 
\end{align*}
The coefficients (\ref{Coeff_a}) and (\ref{Coeff_b}) can now be computed. It follows that:
$$
A_1=v_{3}
\left[\frac{\mu}{\nu_{0}+\mu}(\epsilon u_{3}+u_{4})\right],
$$
and
$$
A_2=-v_{3}\frac{\beta^*\mu}{\Lambda(\nu_{0}+\mu)}
\left[\frac{b\Lambda}{\mu(\nu_{0}+\mu)}ku_{4} \left(u_{3}\epsilon\mu
+ku_{4}\right)u_{4}\mu+2u_{3}^2\epsilon\mu+u_{3}u_{4}(\epsilon+1)\mu
+u_{3}u_{5}\epsilon\mu
+2u_{4}^2\mu
+u_{4}u_{5}\mu\right].
$$
Therefore we have $A_1>0$ and $A_2<0$ and the system undergoes a forward bifurcation. In other words, when $\mathcal{R}_{0}$  crosses the threshold value $\mathcal{R}_{0}=1$ from left to right, the disease free equilibrium changes its stability from  asymptotically stable to unstable and a positive endemic equilibrium $E^*$ appears and this last is locally asymptotically stable, at least for values of $\mathcal{R}_{0}$ close to 1.
\end{proof}

\section{Numerical Investigations}
\label{sec:sens}
In this section we will assess how the endemic equilibrium and the dynamics of the model depend on some key behavioral and epidemiological parameters.\\
In Figure \ref{IeqKB} we consider how the steady state number of infectious at the endemic equilibrium, $I^*$, depends on the normalized information coverage $h$ (whose range is $(0,1)$) for the following values of the sensitivity slope: $b =1000$ (black solid line), and $b= 2000$ (red dashed line). Not surprisingly, $I^*$ is decreasing with increasing both $h$ and $b$. \\
\begin{figure}[t!]
	\centering
		\includegraphics[width=0.65\textwidth]{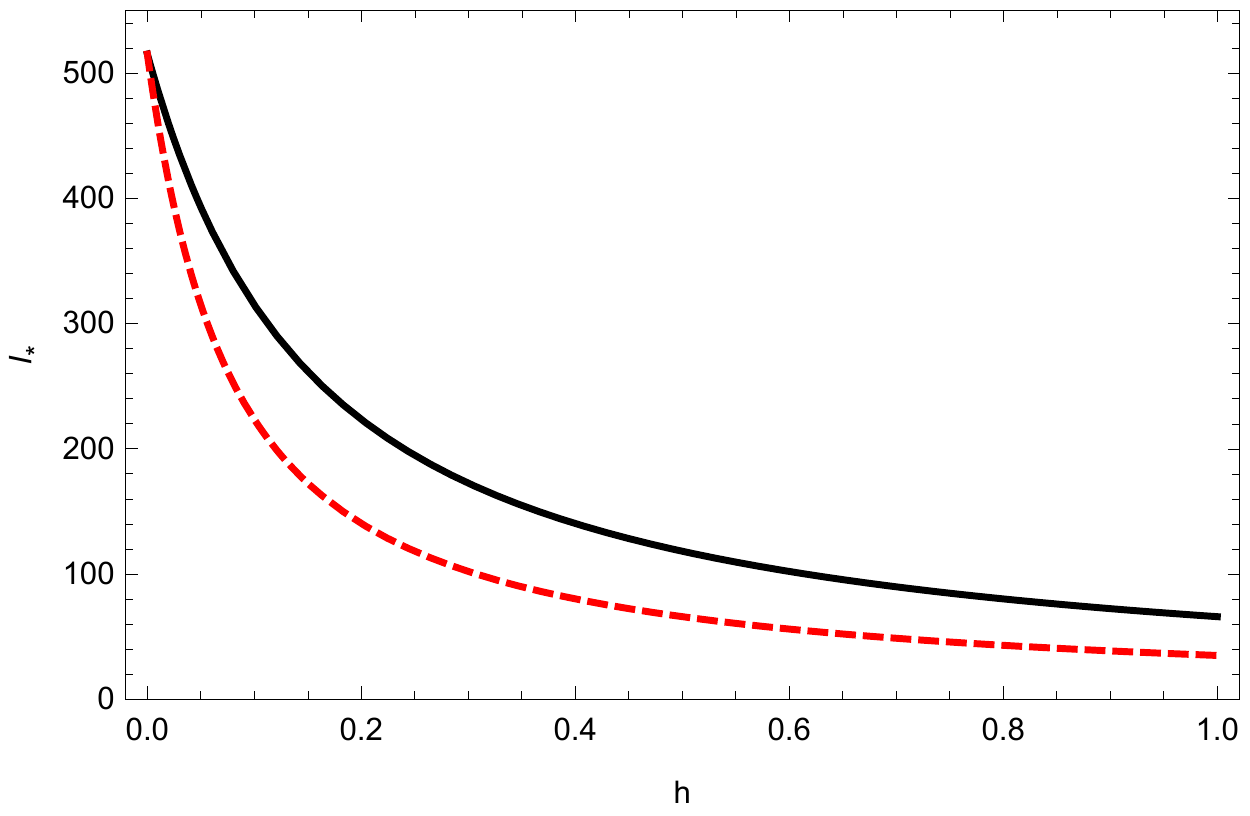}
	\caption{Endemic equilibrium for $I$ vs. normalized information coverage $h$. Solid black line: $b=1000$, dashed red line: $b=2000$. Other parameters: as in the Table \ref{parameters2}.}
	\label{IeqKB}
\end{figure}

As far as the simulations of the spread and control of the disease are concerned, we will consider the following initial conditions:\\
\textit{i)} before the introduction of the disease the population was at its steady state, thus $N(0^-)= \tilde{N}$;\\
\textit{ii)} one single infectious subject enters into the population: $I(0)=1$;\\
\textit{iii)} No carriers and removed subjects are present at $t=0$: $C(0)=0$ and $R(0)=0$;\\
\textit{iv)}  a relatively small fraction of the population has been vaccinated: $V(0) = 0.1\tilde{N}$. It follows that $S(0)= 0.9 \tilde{N}-1$.

\smallskip

\noindent We first consider the epidemic scenario, by simulating the system from the initial time to one year later ($t \in [0,1]$) as illustrated in Figure \ref{CI1yT} and Figure \ref{Pop1yT}. Namely, in the left panel of Figure \ref{CI1yT} we show the time--course of both $C$ and $I$ in the case where there is no behavioral dependence of the vaccination rate, i.e. $b=0$, which implies that  $\nu(M)=\nu_0$. In both the central and right panels of Figure \ref{CI1yT}, where we set $b=2000$, we show the cases corresponding to $T = 10$ $days$ and $T = 120$ $ days$, respectively. For $T=10$ we can see that the prompt response induces a remarkable decrease of the epidemic peak, which is no more observed for $T = 120$  $days$. The time course of the total population $N(t)$ during the first year is shown in Figure \ref{Pop1yT}. 
\begin{figure}[ht]
	\centering
		\includegraphics[width=0.32\textwidth]{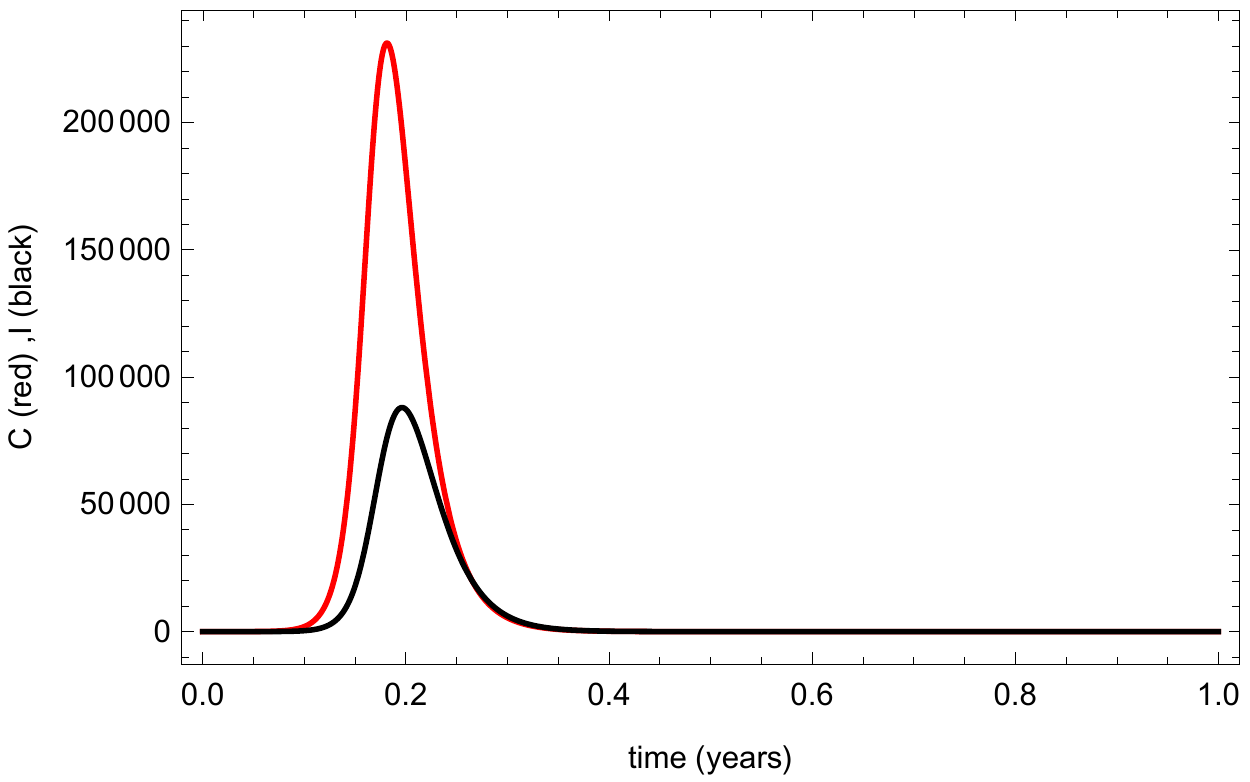}
		\includegraphics[width=0.32\textwidth]{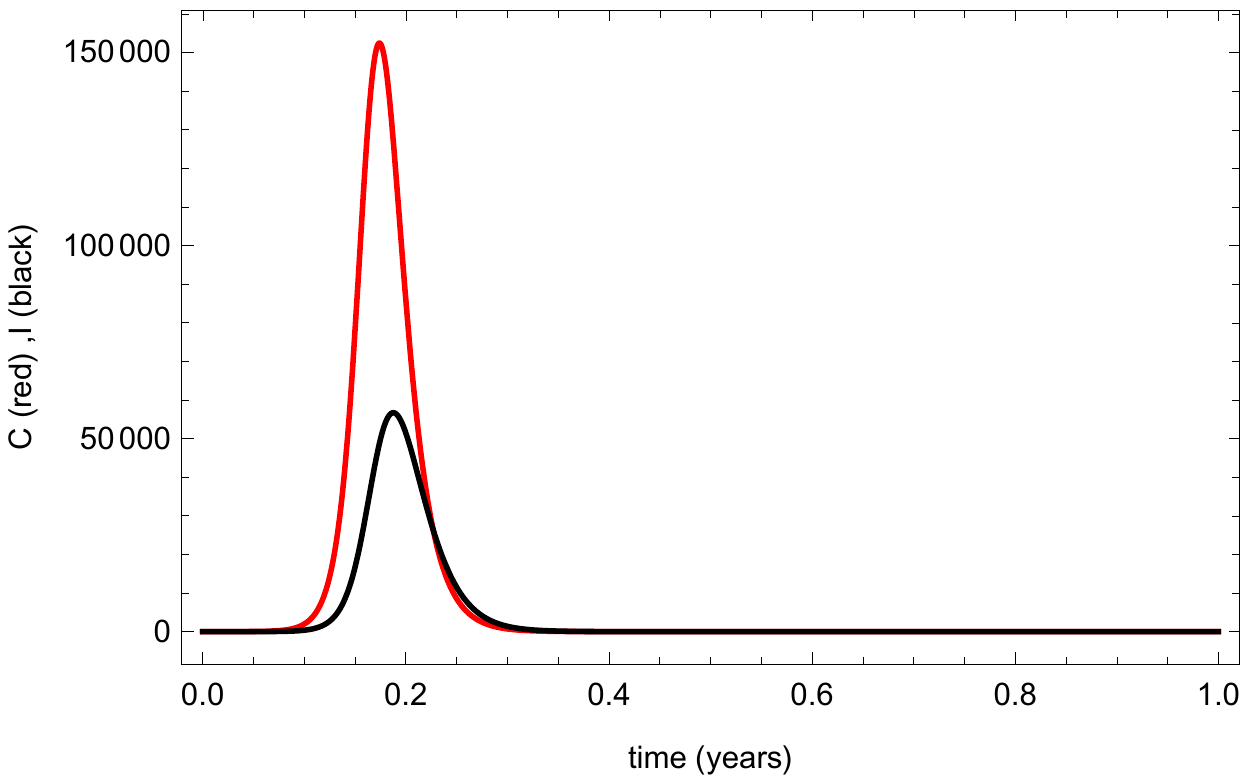}		
		\includegraphics[width=0.32\textwidth]{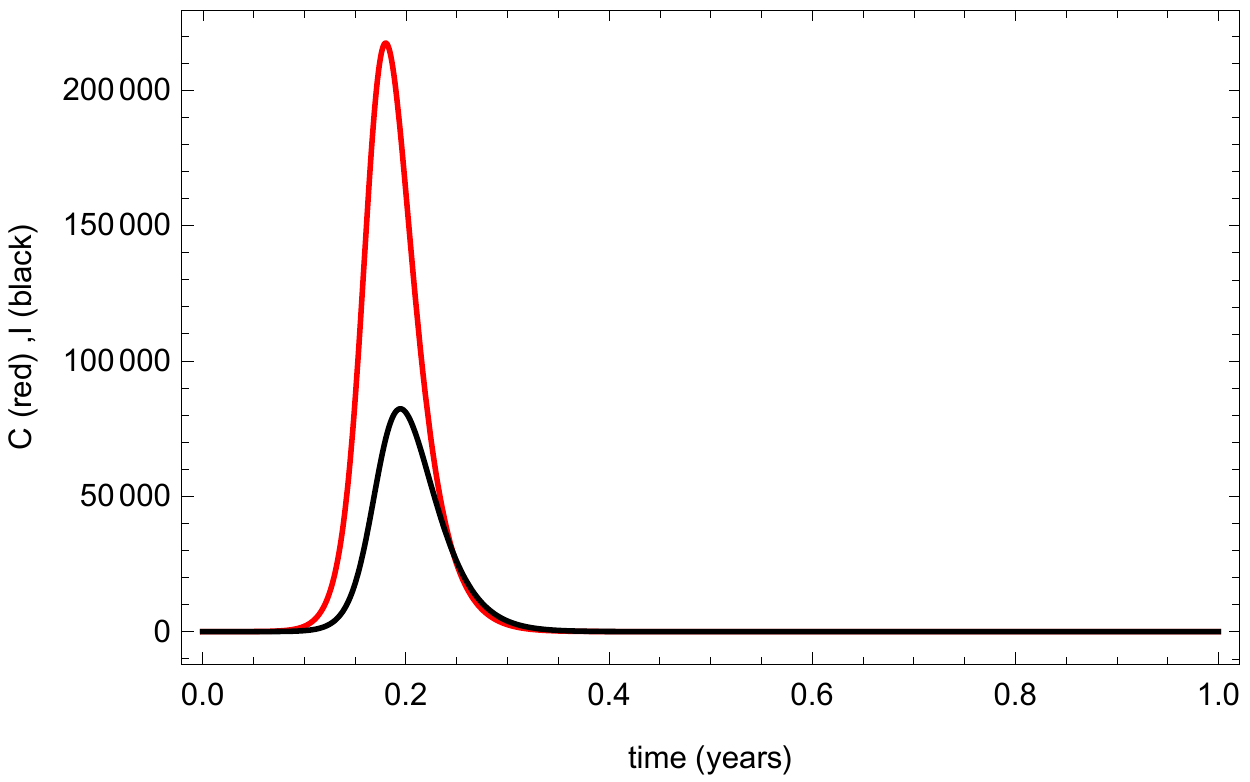}
	\caption{Impact of behavioral terms on the epidemics outbreak during the first year. Time courses of $C(t)$ (red) and of $I(t)$ (black). Left panel: absence of behavioral components in the vaccination ($b=0$); central and left panel: presence of behavioral components in the vaccination ($b=1000$). Central panel: $T=10$ $days$; right panel: $T=120$ $days$. In all panels $\sigma =26$. Other parameters as in the Table \ref{parameters2}}
	\label{CI1yT}
\end{figure}

\begin{figure}[ht]
	\centering
		\includegraphics[width=0.32\textwidth]{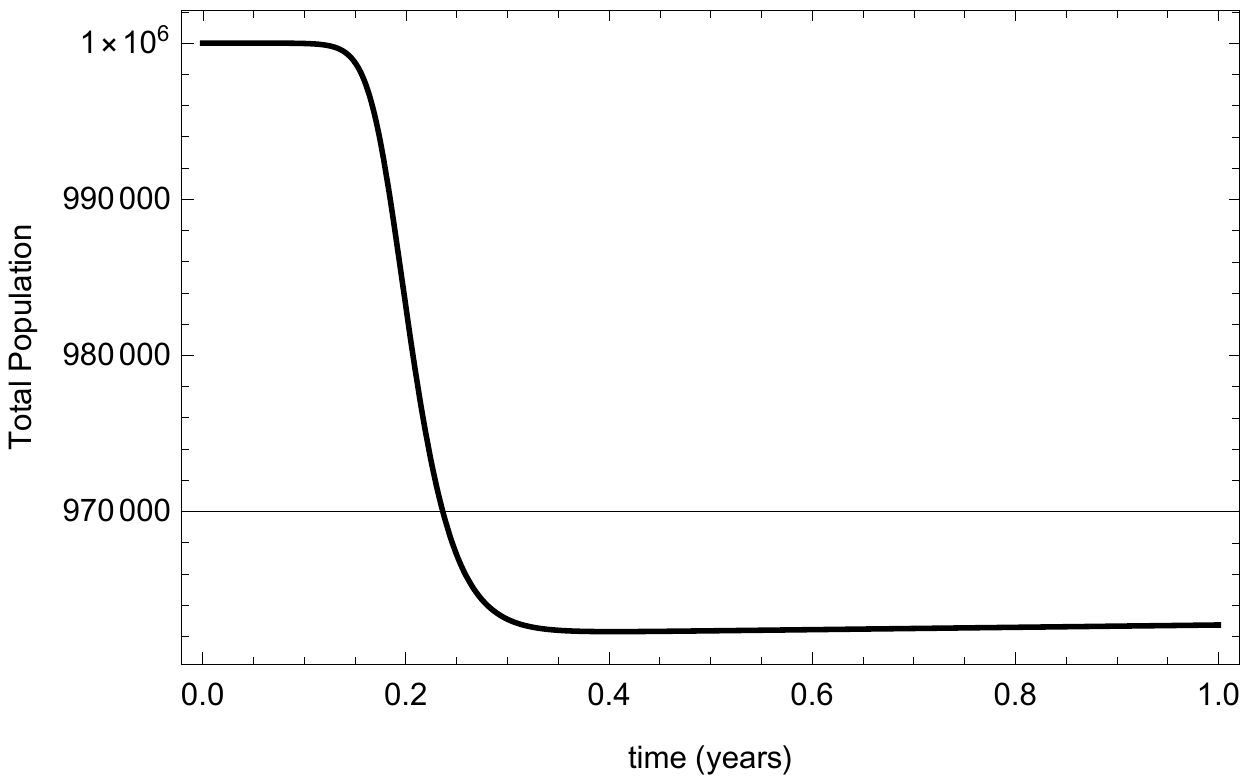}
		\includegraphics[width=0.32\textwidth]{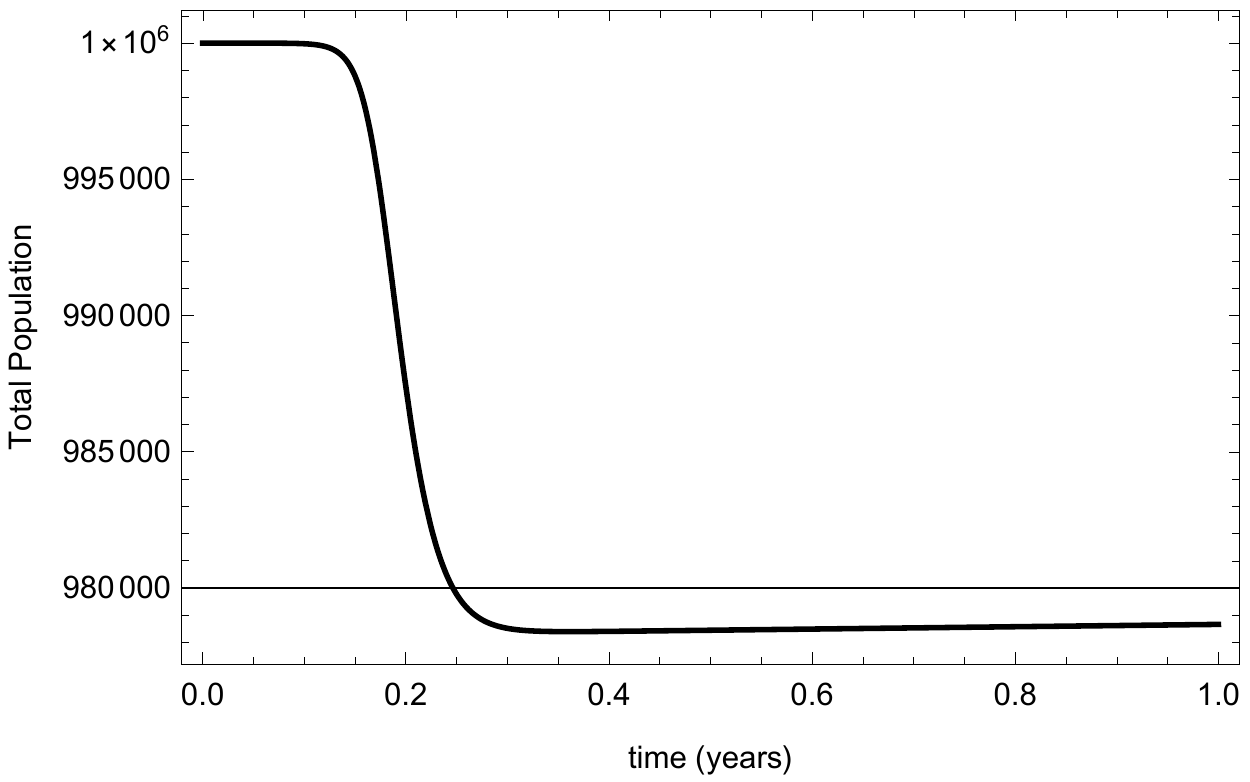}
		\includegraphics[width=0.32\textwidth]{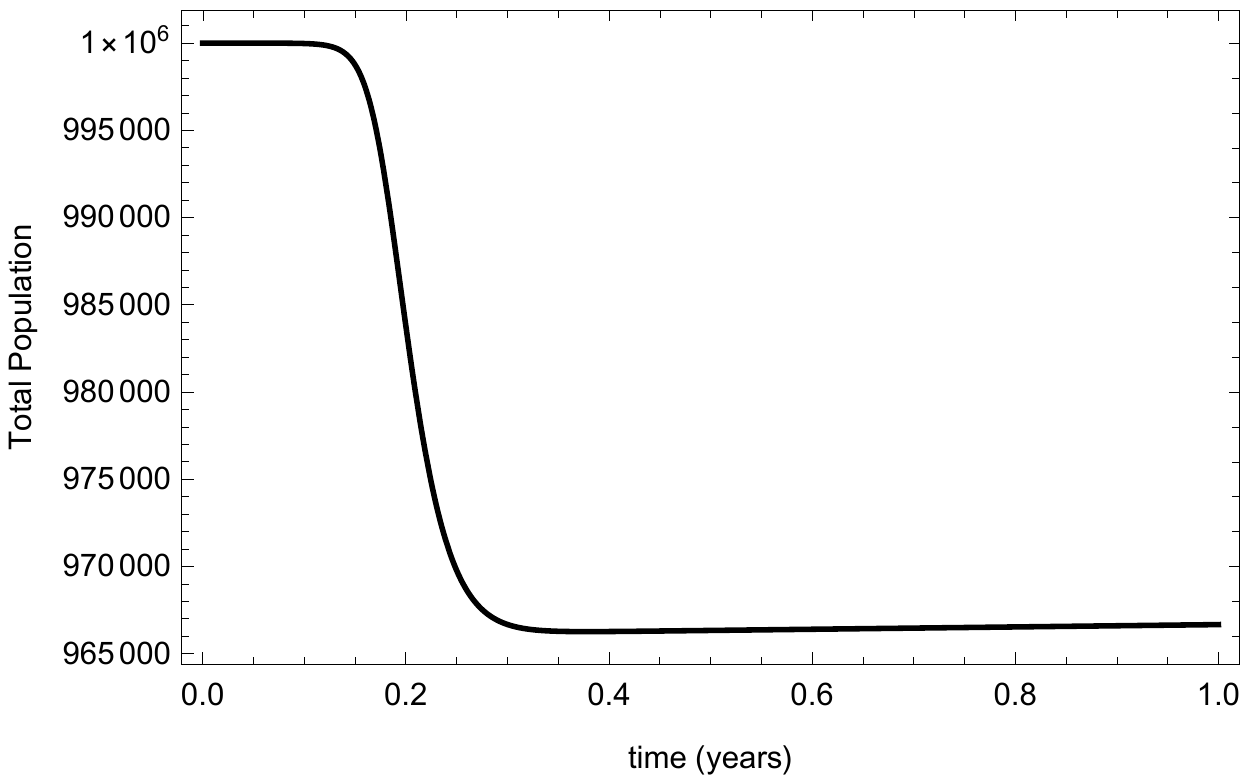}
	\caption{Impact of behavioral terms on the total population $N(t)$ during the epidemics outbreak. Left panel: absence of behavioral components in the vaccination ($b=0$); central and left panel: presence of behavioral components in the vaccination ($b=1000$). Central panel: $T=10$ $days$; right panel: $T=120$ $days$. In all panels $\sigma =26$. Other parameters as in the Table \ref{parameters2}.}
	\label{Pop1yT}
\end{figure}

\smallskip

An analytical assessment of the onset of bifurcations is hard to obtain. Thus, given the parameter $a$, we computed the eigenvalues $\lambda_i(a)$. The maximum value of the real part of the $\lambda_i(1/T)$ is then plotted in Figure \ref{bifudiag} (left panel). As it can be seen, there is an interval $T \in (T_{min},T_{max}) \approx (24,1645)$ $days$ where the endemic equilibrium is unstable and steady state oscillations (not necessarily periodic) can take place. In other words, the model predict the onset of recurrent epidemics. Note that for $T \approx 220$ $days$ the largest positive real part of the eigenvalues reaches its maximum. In Figure \ref{oscillations} we plotted the corresponding steady state oscillations for, respectively, $T=120$ $days$ (left panel) and for $T=220$ $days$ (right panel). As one can see the model predicts recurrent epidemics. Observe that for $T= 220$ $days$ the epidemics have larger peak and period than for $T= 120$ $days$.
\begin{figure}[ht]
	\centering
		\includegraphics[width=8cm,height=5cm]{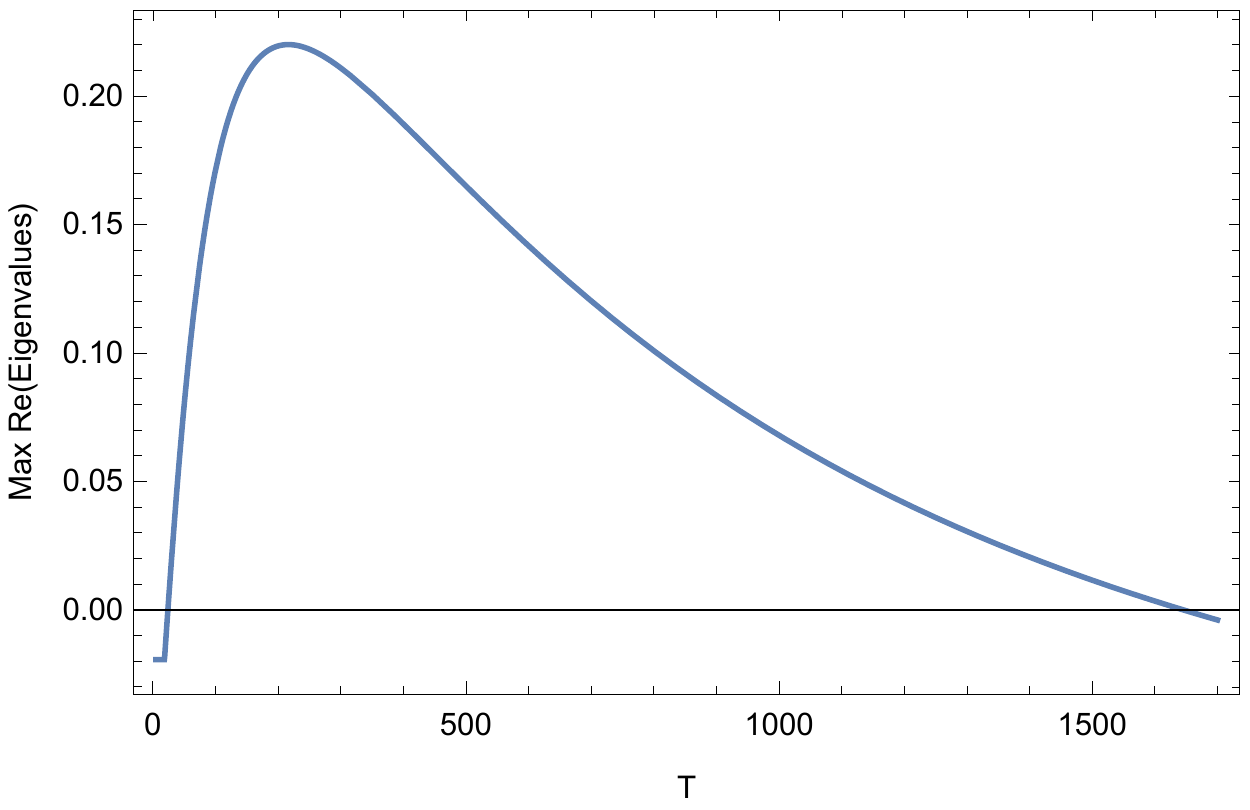}
		\includegraphics[width=8cm,height=5cm,keepaspectratio]{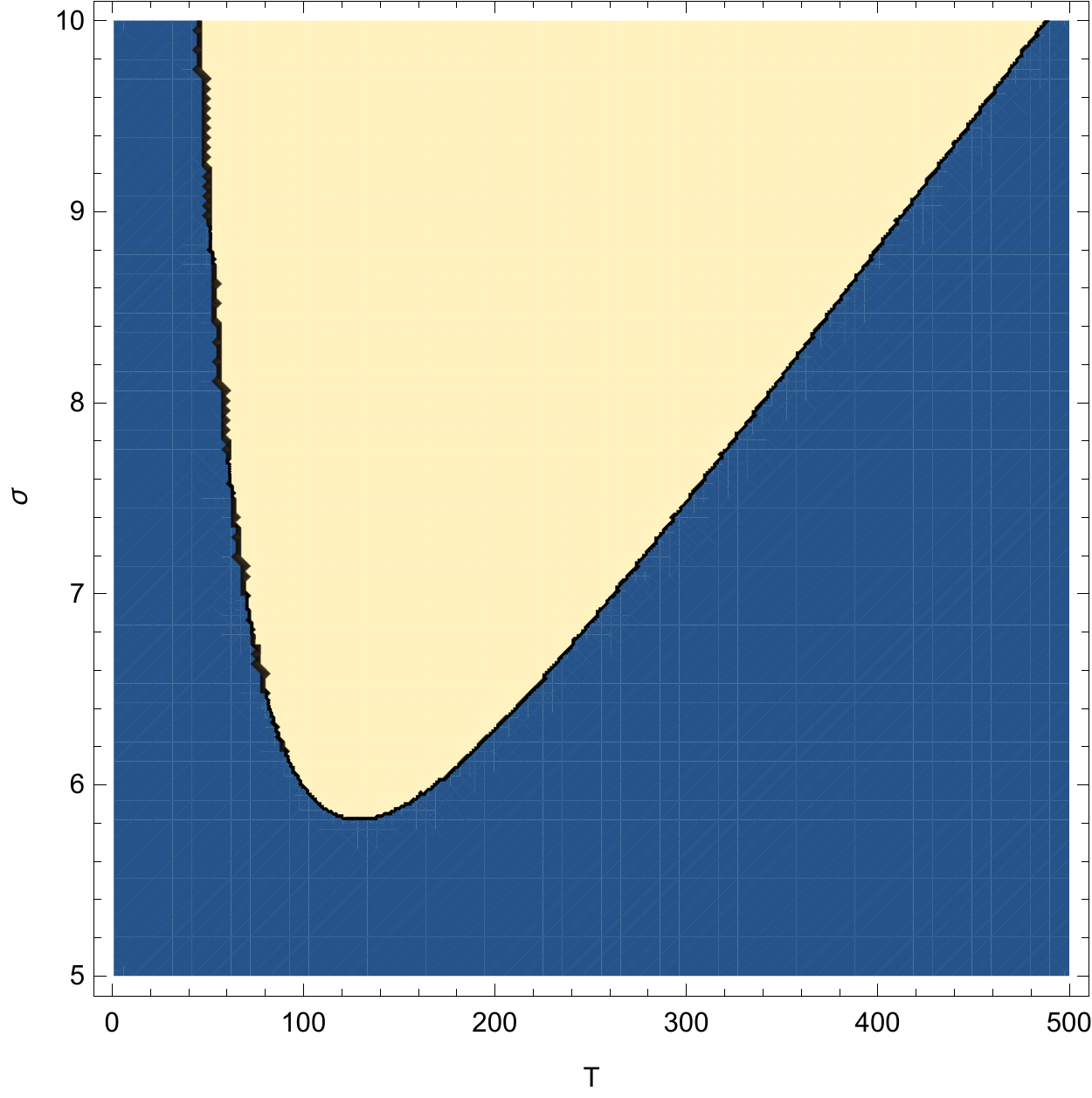}
	\caption{Bifurcation diagrams at the endemic equilibrium. Left panel: one bifurcation parameter, $T$, and $\sigma=26$. Right panel: two bifurcation parameters: $(T,\sigma)$. Other parameters as indicated in the Table \ref{parameters2}. The region of instability is in light color, and the local stability region is dark.}
	\label{bifudiag}
\end{figure}

\begin{figure}[ht]
	\centering
	\includegraphics[width=0.48\textwidth]{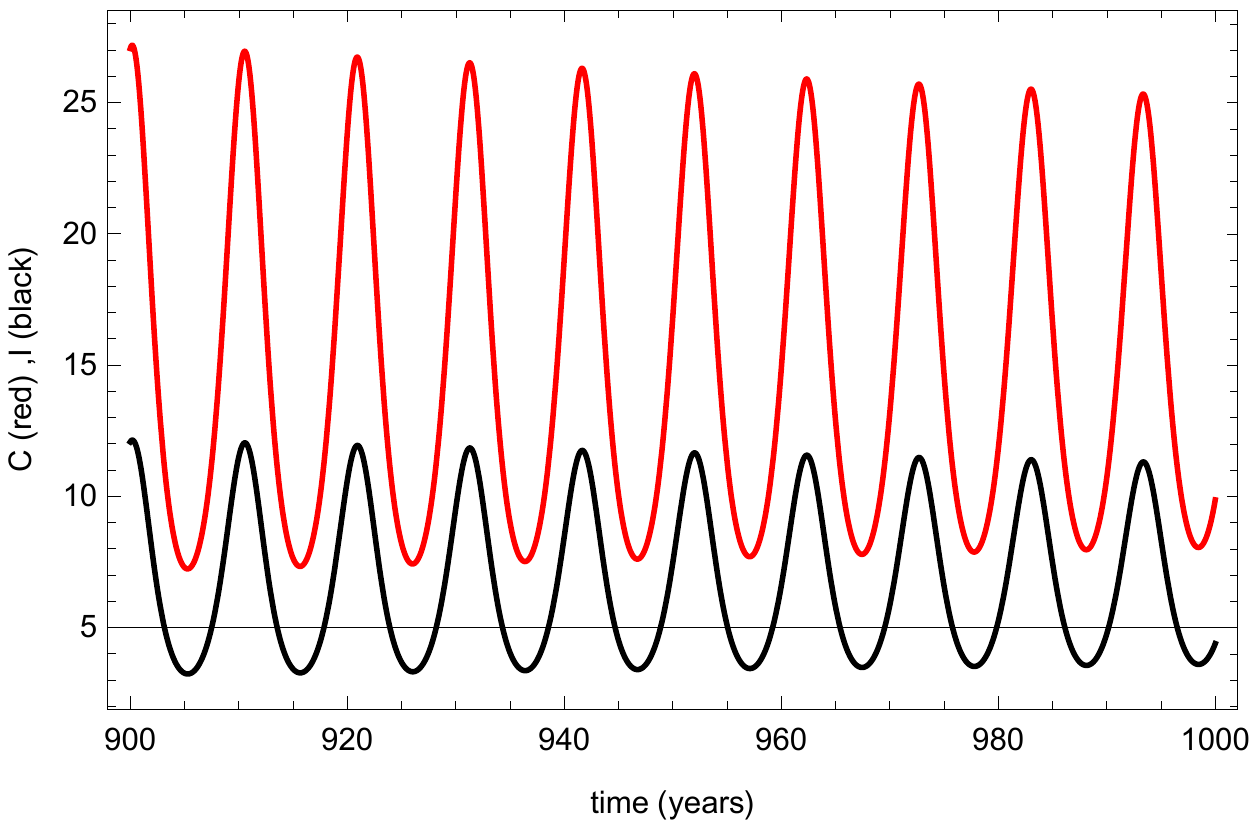}
		\includegraphics[width=0.48\textwidth]{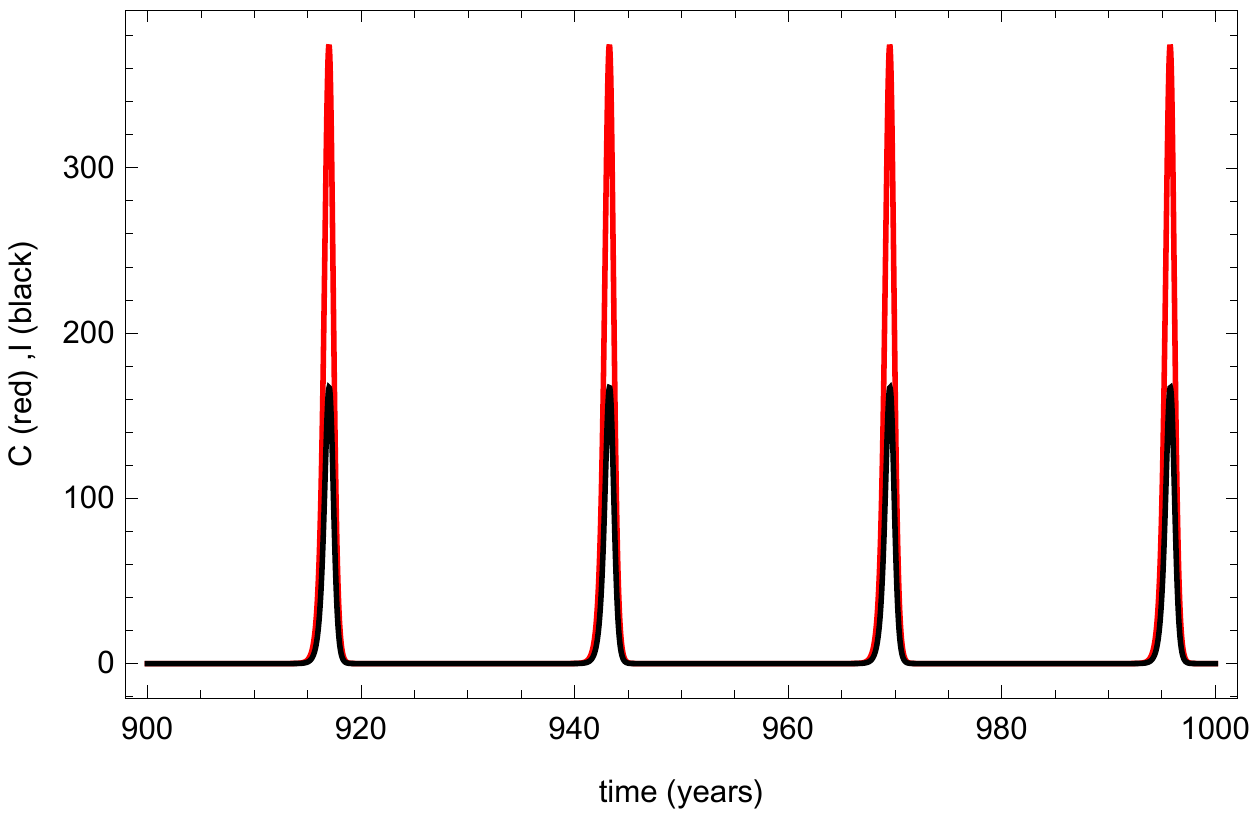}
	\caption{Steady state oscillations. Left panel: $T=120 days$; right panel: $T = 220 days$. Red lines: $C(t)$, black lines: $I(t)$. Other parameters as in Table \ref{parameters2}.}
	\label{oscillations}
\end{figure}

All the above simulations were performed by assuming $\delta \approx \sigma$ (see Table \ref{parameters2}), i.e. the recovery rate of carriers is equal to the rate of moving from carriers to infected, which seems reasonable. This however cause a large infectious peak, as in the models considered by Irving et al. \cite{IBC2012}, where - however - vaccination was not considered. The presence of vaccination and of behavior--dependent increase of the vaccination rate, of course, partially mitigate this phenomenon. Irving and coauthors noticed in \cite{IBC2012} that in their model lower peaks of $I(t)$ are obtained by assuming $\delta >> \sigma$. Since $P_R=\delta/(\delta+\sigma)$ is the probability that an individual going out of the $C$ compartment enters in the $R$ compartment, the above--mentioned  hypotheses is equivalent to assume that $p_R \approx 1$, i.e. that the vast majority of carriers do not become infectious. To compare, the assumption $\delta=\sigma$ correspond to $p_R=0.5$, i.e. that half of carriers become infectious (apart those dying for natural causes, of course).\\
Based on the above considerations, we computed a two-parameters bifurcation diagram (BBD) where not only $T$ but also $\sigma$ is considered as bifurcation parameter. The BBD is shown in  Figure \ref{bifudiag} (right panel), where the region of instability is in light color, and the local stability region is dark. Our numerical computations showed that: \textit{i)} under a threshold value $\sigma = \sigma_m \approx 4.9$ the system is LAS; \textit{ii)} for  $\sigma \in (\sigma_m,26)$ the endemic equilibrium is unstable in a window of values of $T$: 
$T \in (T_{min}(\sigma),T_{max}(\sigma))$ whose size $A=T_{max}(\sigma)-T_{min}(\sigma)$ increases with $\sigma$.\\
To show the impact of $\sigma$ in the first phase of epidemics, we performed alternative simulations by setting a lower value of $p_R$ compared to $p_R=0.5$. Namely we set  $\delta = 10 \sigma$, which implies that $p_R=1/9$. The simulation is reported in Figure \ref{SimuSigma2p6}, where one can observe that the epidemic peak is reduced when compared to the case shown in Figure \ref{CI1yT}.\\
 \begin{figure}[ht]
	 \centering
		 \includegraphics[width=0.48\textwidth]{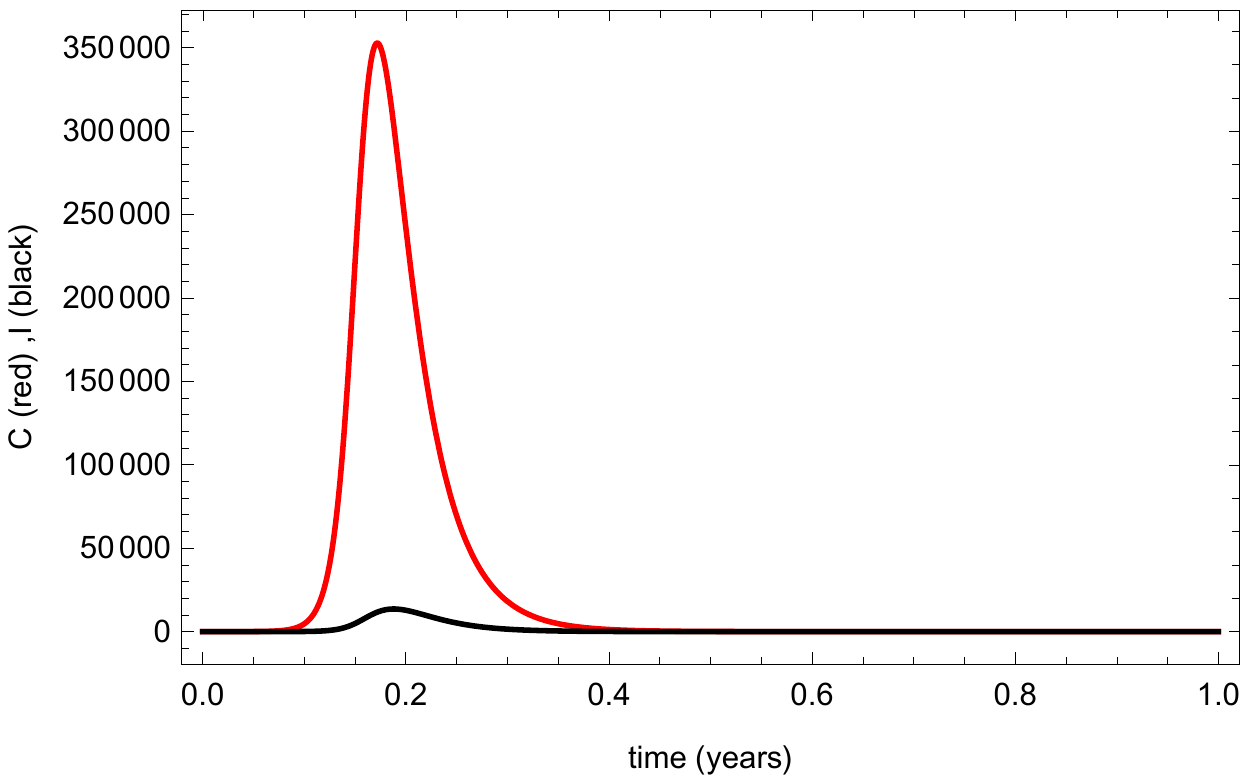}
	 \includegraphics[width=0.48\textwidth]{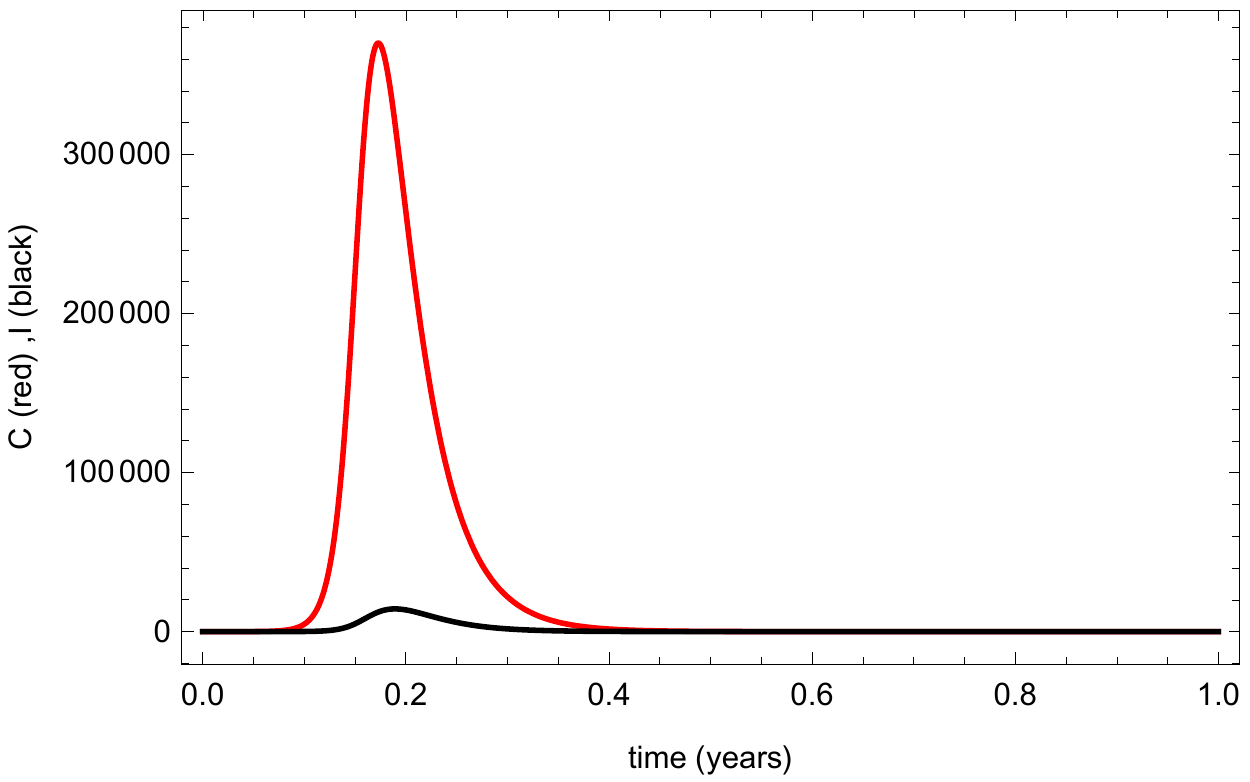}
	 \caption{Simulations of the first year of the spread of the disease, with $\sigma =2.6 years^{-1}$. Left panel: $T=10 days$, right panel: $T=120 days$. Other parameters as in Table \ref{parameters2}.}
	 \label{SimuSigma2p6}
 \end{figure}

%%%%%%%%%%%%%%%%%%%%%%%%%%%%%%%%%%%%%%%%%%%%%%%%%
\section{Conclusions }
\label{sec: concl}

As far as we know, the effects of information-dependent behavior on meningitis transmission has not been studied before in the context of behavioral epidemiology \cite{MD2013}. Here, we propose a variant of a well established meningitis model, where
an information--dependent vaccination behavior is explicitly taken into account. This is done by considering the information index 
as early proposed in \cite{DMS2007}.

The qualitative analysis of the resulting \emph{behavioral change model} is based on stability and bifurcation theory. These results extend (in some way) the ones obtained in \cite{Blyuss2016,IBC2012,Yaleu2017}, in that: \textit{i)} we included behavioral effects; \textit{ii)} we get a rigorous results concerning the transcritical bifurcation taking place at the threshold $\mathcal{R}_{0}=1$.

The BRN of the model does not depend on the information-related parameters, thus the GAS result concerning the DFE can be read (similarly to the case of the SIR model with behavioral vaccination investigated in \cite{DMS2007}) as an \textit{eradication impossible} result \cite{MD2013,WBB2016}. Indeed, to guarantee the disease eradication the baseline behavior--independent vaccination rate ought to be as great as to guarantee the eradication in absence of the behavioral effects in the vaccination rate. This is, of course, unlikely \cite{MD2013,WBB2016}.

As far as the information effects are concerned, our simulations suggests that, of course, it contributes to reduce the epidemic peak. However, this holds in the case where $T$ is medium--small. At the endemic state, the equilibrium size of the compartment of infectious $I^*$ is a decreasing function of the information coverage and of $b$.

The delay $T$ is able to destabilize the endemic equilibrium by inducing steady state recurrent epidemics. For $\sigma=26$ we found that the destabilization is observed for a windows of values of $T$.

Irving and coworkers \cite{IBC2012} stressed that the parameter $\sigma$, which is the rate of the transfer from the carrier to the infectious state, has major impact in absence of vaccination. We found that $\sigma$ remains of major relevance also in presence of behavior-dependent vaccinations, with specific effects. It does not only deeply modify the epidemic outbreaks (by reducing the peak of $I$ and increasing that of $C$), as it is intuitive since it rules how many carriers do become infectious, but it also significantly impact on the onset of recurrent epidemics. Indeed, our numerical analysis suggest that there is a critical value $\sigma_m$ such that for $\sigma< \sigma_m$ no oscillatory solutions occur, and the size of the above--mentioned `instability window' for $T$ increases with $\sigma$. In other words, there is a nonlinear interplay between the epidemics-related parameter $\sigma$ and the behaviour--related parameter $T$.

This study is very preliminary and specific assumptions have been made to make model (\ref{system2}) more tractable. Furthermore, the obtained qualitative results need to be validated by field data, especially as far as the behavioral aspects are concerned. Future investigations will mainly concern three areas: i) the possible interplay between the important seasonal changes of $\beta$ and also of $\sigma$ stressed in \cite{IBC2012} with the behavioral components we introduced here; ii) the effect of more realistic  non--Erlangian kernels (as the \textit{acquisition--fading kernel} introduced in \cite{domaLibro}); iii) the optimal control of public health strategies aimed at increasing vaccine acceptance.

\vspace{0.5cm}

\textbf{Acknowledgements:} {\small The present work has been performed under the auspices of the Italian National Group for the Mathematical Physics (GNFM) of National Institute for Advanced Mathematics (INdAM). The authors B.B., S.M.K. and  Y.H.W.  gratefully acknowledge the University of Naples Federico II that supported their research through the project entitled `{\em Advanced dynamical systems for the analysis and control of transmission of infectious diseases}' in the framework of `NAASCO', bilateral agreement of scientific cooperation between the University of Naples Federico II and the Addis Ababa University, 2016-2021. B.B. is grateful to the organizing committee of the conference `Dynamical Systems Applied to Biology and Natural Sciences (DSABNS)' for the kind invitation and hospitality (Naples, Italy, February 3-6, 2019). S. M. K. is supported by Botswana International University of Science and Technology (BIUST) through the BIUST Initiation grant.}

%\end{linenumbers}
\end{document}